\documentclass[openacc]{rsproca_new}

\newtheorem{theorem}{\bf Theorem}[section]
\newtheorem{corollary}[theorem]{\bf Corollary}
\newtheorem{lemma}[theorem]{\bf Lemma}
\newtheorem{conjecture}[theorem]{\bf Conjecture}
\newtheorem{definition}[theorem]{\bf Definition}
\newtheorem{remark}[theorem]{\bf Remark}

\usepackage{physics}
\usepackage{braket}
\DeclareMathOperator{\image}{Im}
\usepackage{dsfont}

\begin{document}

\title{An algorithm to explore entanglement in small systems}

\author{R. Reuvers$^{1}$}

\address{$^{1}$Department of Applied Mathematics and Theoretical Physics (DAMTP), Centre for Mathematical Sciences, University of Cambridge, Wilberforce Road, Cambridge CB3 0WA, United Kingdom}

\subject{Quantum physics}


\corres{Robin Reuvers\\
\email{r.reuvers@damtp.cam.ac.uk}}

\begin{abstract}
A quantum state's entanglement across a bipartite cut can be quantified with entanglement entropy or, more generally, Schmidt norms. Using only Schmidt decompositions, we present a simple iterative algorithm to maximize Schmidt norms. Depending on the choice of norm, the optimizing states maximize or minimize entanglement, possibly across several bipartite cuts at the same time and possibly only among states in a specified subspace.  

Recognizing that convergence but not success is certain, we use the algorithm to explore topics ranging from fermionic reduced density matrices and varieties of pure quantum states to absolutely maximally entangled states and minimal output entropy of channels.
\end{abstract}


\begin{fmtext}
\section{Introduction}
As a consequence of the singular-value decomposition, a normalized state $\ket{\psi}$ in a bipartite Hilbert space $\mathcal{H}_A\otimes\mathcal{H}_B$ with finite dimensions $d_A, d_B$ and $n_s:=\min(d_A,d_B)$ can be written as
\begin{equation}
\label{Schmidt}
\ket{\psi}=\sum^{n_s}_{i=1} \lambda^\psi_i \ket{\psi^i_A} \otimes \ket{\psi^i_B},
\end{equation}
with \textit{Schmidt coefficients} $\lambda^\psi_1\geq\dots\geq\lambda^\psi_{n_s}\geq0$ satisfying $\sum_i(\lambda^\psi_i)^2=1$ and \textit{Schmidt vectors} $\ket{\psi^i_A}\in\mathcal{H}_A$, $\ket{\psi^i_B}\in\mathcal{H}_B$. 
The sequence of Schmidt coefficients and its related \textit{entanglement spectrum} $\xi_i=-\log(\lambda_i^2)$ characterize entanglement \cite{Ekert} and aid the study of condensed matter systems \cite{LiHaldane}. Often, such systems satisfy symmetry constraints that are imprinted on the Schmidt coefficients---think of the Haldane phase of the $S=1$ spin chain \cite{Pollmann} and symmetries for bosons and fermions. 
\end{fmtext}
\maketitle

The goal of this paper is to point out that the structure of the Schmidt decomposition allows for an easy iteration with which we can explore some of the constraints imposed by symmetries.
Concretely, the procedure attempts to maximize norms (see \cite{Johnston1} for $p=2$), that we refer to as \textit{Schmidt norms},
\begin{equation}
\label{Q}
\|\ket\psi\|_{p,k}:=\left(\sum^k_{i=1} (\lambda^\psi_i)^{p}\right)^{1/p}\ \ \ \ \ \ \ \ \ \ \ \ \ (p\geq1,\ k\leq n_s),
\end{equation}
over pure states $\ket{\psi}$ in a specified subspace $U\subset\mathcal{H}_A\otimes\mathcal{H}_B$. Though they may seem rather abstract, these norms are relevant to a wide range of problems. We believe this can most effectively be demonstrated with examples and shall do so in Section \ref{apps}.

The first sections of this paper are fully general: Section \ref{main} states the algorithm and its goal in the most general terms; Section \ref{moregoal} motivates the relevance of the Schmidt norms by linking them to entanglement; Section \ref{basics} contains the basic properties of the algorithm. We then move on to specific examples in Section \ref{apps}. These demonstrate the scope of the algorithm and motivate why it is often interesting to maximize the norms \eqref{Q}: Section \ref{fermions} looks at fermionic antisymmetry; later sections discuss more mathematical applications such as absolutely maximally entangled states and perfect tensors (\ref{AME}), varieties of pure quantum states (\ref{varieties}) and minimal output entropy of channels (\ref{channels}). We hope that the reader will select one or more of these topics to see what the algorithm can do in a concrete setting---we refer to the conclusion (Section \ref{concl}) for a summary. 

\section{Set-up and algorithm}
\label{main}
This section states the algorithm and its goal.
\subsection{Goal: maximizing Schmidt norms}
\textit{Given a subspace $U\subset\mathcal{H}_A\otimes\mathcal{H}_B$, an integer $k\geq1$ and $p\geq1$, find approximate maximizers of}
\begin{equation}
\label{goal}
\sup_{\substack{\ket{\psi}\in U \\ \|\ket\psi\|=1}} \|\ket\psi\|_{p,k} =\sup_{\substack{\ket{\psi}\in U \\ \|\ket\psi\|=1}}\  \left(\sum^k_{i=1} (\lambda^\psi_i)^{p}\right)^{1/p}.
\end{equation}

Note that this is not a convex optimization problem because we are maximizing, rather than minimizing, a (convex) norm. We refer to these norms as \textit{Schmidt norms}, see Section \ref{Schmidtnormssect} for basic properties. The vector norm $\|.\|$ equals $\|.\|_{2,n_s}$.

\subsection{The algorithm}
\label{algo}
Let $P$ be the orthogonal projection onto the subspace $U$, i.e.\ $P^2=P$, $P^\dagger=P$ and $\image(P)=U$.
\textit{\begin{enumerate}
\item Pick any initial non-zero $\ket{\psi}\in U$.
\item Schmidt decompose $\ket{\psi}$, see \eqref{Schmidt}.
\item Define the normalized vector
\begin{equation}
\label{iteration}
\ket{\phi}:=\frac{P\left(\sum^k_{i=1} (\lambda^\psi_i)^{p-1} \ket{\psi^i_A} \otimes \ket{\psi^i_B}\right)}{\left\|P\left(\sum^k_{i=1} (\lambda^\psi_i)^{p-1} \ket{\psi^i_A} \otimes \ket{\psi^i_B}\right)\right\|}.
\end{equation}
\item Redefine $\ket{\psi}:=\ket{\phi}$ and repeat from $(ii)$.
\end{enumerate}}

In words, the idea is to use repeated Schmidt decompositions while prioritizing the Schmidt coefficients according to the power $p$ and forcing the state to be in the desired subspace.

In Section \ref{secconv}, we prove that $\|\ket\psi\|_{p,k}$ increases with each step and that it converges to a fixed point that may or may not be the global maximum. Section \ref{genr} describes a generalization to several bipartite cuts. The other subsections of \ref{basics} discuss basic properties and can mostly be read independently. It is also possible to skip ahead to the applications in Section \ref{apps}.

\section{Schmidt norms and entanglement}
\label{moregoal}
This section provides some background on the Schmidt norms and explains their relation with entanglement.

\subsection{Schmidt norms}
\label{Schmidtnormssect}
The norms \eqref{Q} for $p=2$ were first studied in quantum information theory in \cite{Johnston1}. More generally, a Schmidt norm \eqref{Q} is an operator norm if $\ket\psi$ in \eqref{Schmidt} is regarded as a map $\mathcal{H}_B\to\mathcal{H}_A$ defined by $\sum_{i} \lambda^\psi_i \ket{\psi^i_A}\bra{\psi^i_B}$. It is the \textit{Schatten $p$-norm} if $k=n_s$ and the \textit{Ky Fan $k$-norm} if $p=1$ \cite{Bhatia}. Of course, we can use well-known properties and results for these norms, such as the one below.

\begin{lemma}[for example \cite{Bhatia}, Section IV.2]
For $k\leq n_s$ and $p\geq1$, the quantities \eqref{Q} are norms that are invariant under unitaries acting on $\mathcal{H}_A$ and $\mathcal{H}_B$.
\end{lemma}
\begin{proof}
It will be convenient to prove the following variational characterization of the Schmidt coefficients. For a decreasing sequence $c_1\geq \dots \geq c_k \geq0$, we claim that 
\begin{equation}
\label{varchar}
\sum^k_{i=1}c_i\lambda^\psi_i=\sup_{\substack{\ket{u_1},\dots,\ket{u_k}\in\mathcal{H}_A \text{\ orthonormal}\\\ket{v_1},\dots,\ket{v_k}\in\mathcal{H}_B \text{\ orthonormal}}}\sum^k_{i=1}c_i\ |\braket{\psi|u_i\otimes v_i}|.
\end{equation}
Inserting the Schmidt vectors of $\ket\psi$, the LHS is clearly lesser than or equal to the RHS. To prove the opposite inequality, let $\ket{u_1},\dots,\ket{u_k}\in\mathcal{H}_A$ orthonormal and $\ket{v_1},\dots,\ket{v_k}\in\mathcal{H}_B$ orthonormal. We have
\begin{equation}
\label{eq22}
\begin{aligned}
\sum^k_{i=1}c_i\ |\braket{\psi|u_i\otimes v_i}|&= \sum^k_{i=1}\sum^{n_s}_{j=1}c_i\lambda^\psi_j |\braket{\psi^j_A|u_i}\braket{\psi^j_B|v_i}|\\
&=\sum_{i,j}\left(\sqrt{c_i\lambda^\psi_j}\ |\braket{\psi^j_A|u_i}|\right)\left(\sqrt{c_i\lambda^\psi_j}\ |\braket{\psi^j_B|v_i}|\right)\\
&\leq\left(\sum_{i,j}c_i\lambda^\psi_j\ |\braket{\psi^j_A|u_i}|^2\right)^{1/2}\left(\sum_{i,j}c_i\lambda^\psi_j\ |\braket{\psi^j_B|v_i}|^2\right)^{1/2}\\
\end{aligned}
\end{equation}
by the Schmidt decomposition \eqref{Schmidt} and the Cauchy--Schwarz inequality.
This is a product of two similar quantities, one of which is
\begin{equation}
\begin{aligned}
\sum_{i,j}c_i\lambda^\psi_j\ |\braket{\psi^j_A|u_i}|^2&=\sum^k_{i=1}c_i\Tr[\ketbra{u_i}\left(\sum^{n_s}_{j=1}\lambda^\psi_j\ket{\psi^j_A}\bra{\psi^j_A}\right)]\\
&=\sum^k_{i=1}(c_i-c_{i+1})\Tr[P_{\text{span}\{\ket{u_1},\dots,\ket{u_i}\}}\left(\sum^{n_s}_{j=1}\lambda^\psi_j\ket{\psi^j_A}\bra{\psi^j_A}\right)],
\end{aligned}
\end{equation}
where $P_{\text{span}\{\ket{u_1},\dots,\ket{u_i}\}}$ is the orthogonal projection onto the span of $\ket{u_1},\dots,\ket{u_i}$ and we have defined $c_{k+1}:=0$. Since $c_i-c_{i+1}\geq0$, we see that the maximum is attained for $\ket{u_i}=\ket{\psi^i_A}$, and \eqref{eq22} implies
\begin{equation}
\begin{aligned}
\sum^k_{i=1}c_i\ |\braket{\psi|u_i\otimes v_i}|&\leq \left(\sum_{i,j}c_i\lambda^\psi_j\ |\braket{\psi^j_A|u_i}|^2\right)^{1/2}\left(\sum_{i,j}c_i\lambda^\psi_j\ |\braket{\psi^j_B|v_i}|^2\right)^{1/2}\\
&\leq\ \sum^{k}_{i=1}c_i\lambda^\psi_i,
\end{aligned}
\end{equation}
which proves the claim.

To prove the statement, the only non-trivial thing to check is the triangle inequality. Ignoring normalization, suppose \eqref{Schmidt} is the Schmidt decomposition of $\ket\psi:=\ket{\psi_1}+\ket{\psi_2}$. We then have
\begin{equation}
\begin{aligned}
\|\ket\psi\|^p_{p,k}&=\left|\sum^k_{i=1} (\lambda^\psi_i)^{p-1} \braket{\psi|\psi^i_A\otimes\psi^i_B}\right|\\
&\leq \sum^k_{i=1} (\lambda^\psi_i)^{p-1}|\braket{\psi_1|\psi^i_A\otimes\psi^i_B}|+\sum^k_{i=1} (\lambda^\psi_i)^{p-1}|\braket{\psi_2|\psi^i_A\otimes\psi^i_B}|.
\end{aligned}
\end{equation}
By the variational expression \eqref{varchar} with $c_i=(\lambda^\psi_i)^{p-1}$ and H\"older's inequality \footnote{For sequences $(x_1,\dots,x_k),(y_1,\dots,y_k)\in\mathbb{C}^k$, and $a,b\in(1,\infty)$ with $1/a+1/b=1$, H\"older's inequality implies $\sum^k_{i=1}|x_iy_i|\leq(\sum^k_{i=1}|x_i|^a)^{1/a}(\sum^k_{i=1}|y_i|^b)^{1/b}$.}, we find
\begin{equation}
\begin{aligned}
\|\ket\psi\|^p_{p,k}&\leq \sum^k_{i=1} (\lambda^\psi_i)^{p-1}\lambda^{\psi_1}_i+\sum^k_{i=1} (\lambda^\psi_i)^{p-1}\lambda^{\psi_2}_i\\
&\leq \|\ket\psi\|^{p-1}_{p,k}(\|\ket{\psi_1}\|_{p,k}+\|\ket{\psi_2}\|_{p,k}),
\end{aligned}
\end{equation}
which implies the triangle inequality.
\end{proof}

\subsection{Entanglement entropies}
\label{EE}
\noindent The reduction of $\ket{\psi}$ in \eqref{Schmidt} to system $A$ is
\begin{equation}
\label{rdm}
\rho^\psi_A:=\Tr_B[\dyad{\psi}{\psi}]=\sum^{n_s}_{i=1} (\lambda^\psi_i)^2 \ketbra{\psi^i_A}{\psi^i_A},
\end{equation}
and similar for the \textit{reduced density matrix} on $B$. This expression is one reason that the Schmidt norms \eqref{Q} are useful. For example,
\begin{equation}
\|\ket\psi\|^p_{p,n_s}=\Tr_A[(\rho^\psi_A)^{p/2}],
\end{equation}
so that the $\alpha$-R\'enyi entropy \cite{Renyi}, defined for a density matrix $\rho$ (for $\alpha\geq0$ and $\alpha\neq1$) as
\begin{equation}
\label{renent}
S_\alpha(\rho):=\frac{1}{1-\alpha}\log(\Tr[\rho^\alpha]),
\end{equation}
becomes
\begin{equation}
\label{renyipsi}
S_\alpha(\rho^\psi_A)=\frac{2\alpha}{1-\alpha}\log(\|\ket\psi\|_{2\alpha,n_s}).
\end{equation}
The \textit{von Neumann entropy} is
\begin{equation}
\label{vnent}
S(\rho):=\Tr[-\rho\log\rho]=\lim_{\alpha\to1}S_\alpha(\rho),
\end{equation}
and $S(\rho^\psi_A)$ is the \textit{entanglement entropy} of $\ket\psi$ across the bipartition. Note that we can approximately minimize this entropy by applying the algorithm to $\alpha\downarrow1$, whereas $\alpha\uparrow1$ gives an approximate maximization. Both optimizers are pure states by convexity of the underlying norm.

\section{Basic Results}
\label{basics}
This section discusses basic properties of the algorithm, such as convergence and generalizations.
\subsection{Convergence and fixed point equation}
\label{secconv}
\begin{theorem}
\label{convergence}
Let $p\geq1$. Under the iteration \eqref{iteration}, the quantity $\|\ket\psi\|_{p,k}$ defined in \eqref{Q} increases and converges.
\end{theorem}
\begin{proof}
We note
\begin{equation}
\label{innerprod}
\|\ket\psi\|_{p,k}^p=\bra{\psi}P\left(\sum^k_{i=1} (\lambda^\psi_i)^{p-1} \ket{\psi^i_A} \otimes \ket{\psi^i_B}\right),
\end{equation}
where the orthogonal projection $P$ was placed in with $P^\dagger=P$ and $P\ket{\psi}=\ket{\psi}$. Consequently, by Cauchy--Schwarz, 
\begin{equation}
\label{increases}
\|\ket\psi\|_{p,k}^p\leq\left\|P\left(\sum^k_{i=1} (\lambda^\psi_i)^{p-1} \ket{\psi^i_A} \otimes \ket{\psi^i_B}\right)\right\|.
\end{equation}
At the same time, by \eqref{iteration} and \eqref{varchar},
\begin{equation}
\label{eq2}
\begin{aligned}
\left\|P\left(\sum^k_{i=1} (\lambda^\psi_i)^{p-1} \ket{\psi^i_A} \otimes \ket{\psi^i_B}\right)\right\|&=
\left|\sum^k_{i=1}(\lambda^\psi_i)^{p-1} \braket{\phi|\psi^i_A \otimes\psi^i_B}\right|\\
&\leq \sum^k_{i=1}(\lambda^\psi_i)^{p-1} |\braket{\phi|\psi^i_A \otimes\psi^i_B}|\\
&\leq \sum^k_{i=1}(\lambda^\psi_i)^{p-1}\lambda^{\phi}_i,
\end{aligned}
\end{equation}
and finally, by H\"older's inequality,
\begin{equation}
\label{eq3}
\sum^k_{i=1}(\lambda^\psi_i)^{p-1}\lambda^{\phi}_i\leq \|\ket\psi\|_{p,k}^{p-1}\|\ket\phi\|_{p,k}.
\end{equation}
Combining \eqref{increases}, \eqref{eq2} and \eqref{eq3}, we see that $\|\ket\psi\|_{p,k}$ increases under the iteration. It is also bounded since $\lambda^\psi_i\leq1$ and so $\|\ket\psi\|^p_{p,k}\leq \|\ket\psi\|^2_{2,n_s}=1$ in all dimensions for $p\geq2$, or alternatively $\|\ket\psi\|^p_{p,k}\leq n_s$ in finite dimensions for $p<2$. A bounded, increasing sequence converges.
\end{proof}

Note that this theorem does not say anything about convergence of the vector. The best we can do in this respect is to use \eqref{iteration}, \eqref{innerprod}, \eqref{eq2} and \eqref{eq3} and note
\begin{equation}
\braket{\psi|\phi}=\frac{\|\ket\psi\|_{p,k}^p}{\left\|P\left(\sum^k_{i=1} (\lambda^\psi_i)^{p-1} \ket{\psi^i_A} \otimes \ket{\psi^i_B}\right)\right\|}\geq\frac{\|\ket\psi\|_{p,k}}{\|\ket\phi\|_{p,k}},
\end{equation}
and so
\begin{equation}
\label{dist}
\|\ket\psi-\ket\phi\|^2\leq2-2\frac{\|\ket\psi\|_{p,k}}{\|\ket\phi\|_{p,k}}.
\end{equation}
Even though the RHS converges upon iteration, this is not quite enough for convergence of the vector because we could in theory move between vectors that all have a large norm $\|\ket\psi\|_{p,k}$. 

Although a better understanding of the iteration would be desirable, in practical applications we can terminate the iteration whenever the increments in $\|.\|_{p,k}$ drop below a set tolerance and the resulting vector can then be used as an output vector with a large norm.

Note that \eqref{dist} implies that a global maximizing vector must be a fixed point of the iteration because the quantity $\|\ket\psi\|_{p,k}$ cannot increase further in that case. Any other fixed points satisfy the following equation.

\begin{corollary}
\label{fixcor}
Fixed points of the iteration \eqref{iteration} satisfy 
\begin{equation}
\label{fixedpoint}
\ket{\psi}:=\frac{P\left(\sum^k_{i=1} (\lambda^\psi_i)^{p-1} \ket{\psi^i_A} \otimes \ket{\psi^i_B}\right)}{\left\|P\left(\sum^k_{i=1} (\lambda^\psi_i)^{p-1} \ket{\psi^i_A} \otimes \ket{\psi^i_B}\right)\right\|},
\end{equation}
and the global maximum is such a fixed point.
\end{corollary}

For general $P$, it is not at all clear how many states satisfy the fixed point equation \eqref{fixedpoint} and how often the process converges to fixed points that are not the global maximum. Luckily, there are many things we can do without running into this problem, see Section \ref{apps}.

\subsection{Generalizations: coefficients and several cuts}
\label{genr}
The iteration \eqref{iteration} owes its efficacy to the structure of the Schmidt decomposition, and it is natural to ask whether the idea can be generalized. We needed $p\geq1$ for H\"older's inequality in the last line of the proof of Theorem \ref{convergence}---indeed the result does not hold for $p<1$. Nevertheless, there are two generalizations that do work.

First, we could simply include some coefficients in the norm \eqref{goal}. It is essential that they are positive and ordered in size.
\begin{corollary}
For an integer $k$, real numbers $c_1\geq\dots\geq c_k\geq0$ and $p\geq1$, the optimization problem
\begin{equation}
\sup_{\substack{\ket{\psi}\in U \\ \|\ket\psi\|=1}}\  \left(\sum^k_{i=1} c_i(\lambda^\psi_i)^{p}\right)^{1/p}
\end{equation}
allows for a version of Theorem \ref{convergence} with an iteration \eqref{iteration} that is adapted identically.
\end{corollary}

Second, we can try to simultaneously maximize across different cuts.
This requires some notational caution: so far we have assumed that we were only studying one cut that was implicit in the definition of the Schmidt norm, so we will avoid talking about norms for now. Assume $\ket\psi\in\mathcal{H}=\mathcal{H}_A\otimes\mathcal{H}_B=\mathcal{H}_A'\otimes\mathcal{H}_B'$, with two corresponding decompositions
\begin{equation}
\label{Schmidt2}
\ket{\psi}=\sum_{i} \lambda^\psi_i \ket{\psi^i_A} \otimes \ket{\psi^i_B}=\sum_{j} \mu^\psi_j \ket{\psi^j_{A'}} \otimes \ket{\psi^j_{B'}}.
\end{equation}

We can maximize sums involving both at the same time: the following statement demonstrates this, and is easy to generalize. 
\begin{corollary}
\label{morecuts}
For $p\geq1$ and $k,l$ integers, the quantity
\begin{equation}
\label{goal2}
\sup_{\substack{\ket{\psi}\in U \\ \|\ket\psi\|=1}}\left[ \sum^k_{i=1} (\lambda^\psi_i)^{p}+\sum^l_{j=1} (\mu^\psi_j)^{p}\right],
\end{equation}
increases and converges under the iteration \eqref{iteration} if $\ket\phi$ is defined as the normalized vector proportional to
\begin{equation}
P\left(\sum^k_{i=1} (\lambda^\psi_i)^{p-1} \ket{\psi^i_A} \otimes \ket{\psi^i_B}+\sum^l_{j=1} (\mu^\psi_j)^{p-1} \ket{\psi^j_{A'}} \otimes \ket{\psi^j_{B'}}\right).
\end{equation}
\end{corollary}

As shown in Section \ref{apps}, this can work very well, but there is a practical objection to considering too many cuts. We mentioned the risk of converging to a fixed point that is not the global maximum when we derived the fixed point equation in Corollary \ref{fixcor}. The adapted algorithm has a similar equation, but this time it contains many more `variables' and it is therefore likely to have more solutions. Hence, at least heuristically, the more cuts, the larger the risk of ending up at unwanted fixed points---something that for example becomes evident in the application in Section \ref{AME}. It would be very interesting to find a solution to this problem, for instance by adding a simulated annealing component to the algorithm.

\subsection{Number of operations}
Recall $d_A:=\dim(\mathcal{H}_A)$ and $d_B:=\dim(\mathcal{H}_B)$.
The most costly operation in the algorithm is acting with the projection $P$ on a vector of length $d_Ad_B$, which requires $O(d_A^2d_B^2)$ operations. In some cases, it can be hard to generate $P$ in an efficient way, so that might also be expensive.

One can also ask how many iterations are required to reach convergence, but in practice this depends heavily on the application, see for example Table \ref{numbers}.

\subsection{Comparison with an iteration suggested by Shor}
The algorithm shares some features with an idea attributed to Shor by Datta and Ruskai \cite{DattaRuskai}. Adapting its purpose from minimizing entanglement entropy to the current goal \eqref{goal} with $k=n_s$ and $p=2\alpha\geq2$, that is $\Tr[(\rho^\psi_A)^\alpha]=\|\psi\|^{2\alpha}_{2\alpha,n_s}$ with $\alpha\geq1$, it goes as follows.
\textit{\begin{enumerate}
\item Pick any initial non-zero $\ket{\psi}\in U$.
\item Calculate the density matrix $P((\rho^\psi_A)^{\alpha-1}\otimes\mathds{1})P$ and determine the eigenvector $\ket\phi$ corresponding to the largest eigenvalue.
\item Redefine $\ket{\psi}:=\ket{\phi}$ and repeat from $(ii)$.
\end{enumerate}}

The norm now increases in a similar way upon iteration (we use that $\ket\psi$ and $\ket\phi$ are eigenvectors of $P$ and H\"older's inequality).
\begin{equation}
\begin{aligned}
\|\ket\psi\|^{2\alpha}_{2\alpha,n_s}&=\Tr_{AB}\left[\dyad{\psi}{\psi}P((\rho^\psi_A)^{\alpha-1}\otimes\mathds{1})P\right]\\
&\leq \Tr_{AB}\left[\dyad{\phi}{\phi}P((\rho^\psi_A)^{\alpha-1}\otimes\mathds{1})P\right]\\
&=\Tr_A\left[(\rho^\psi_A)^{\alpha-1}\rho^\phi_A\right]\\
&\leq \|\ket\psi\|^{2\alpha-2}_{2\alpha,n_s}\|\ket\psi\|^{2}_{2\alpha,n_s}.
\end{aligned}
\end{equation}

The algorithm in Section \ref{algo} only works with vectors, whereas this one uses density matrices. 
That has two consequences: I. its scope is restricted to $p=2\alpha\geq2$ and traces ($k=n_s$) and II. its computational cost is higher because it requires multiplication of square matrices of size $d_Ad_B$, rather than just acting with such a matrix on a vector.

\subsection{The $p=2$ case}
\label{Johnston}
The optimization problem \eqref{goal} for $p=2$ and general $k$ has been studied before \cite{Pankowski, Johnston1}. It can be written as an operator norm of the projection $P$ defined in terms of the vector norm $\|.\|_{2,k}$. This quantity appears on the right-hand side of the lemma below. Note that the condition $\|\ket\psi\|_{2,k}=1$ simply says that the \textit{Schmidt rank}---the number of non-zero Schmidt coefficients---is $k$ or less.
\begin{lemma}[appeared in \cite{Pankowski}]
\label{Panko}
We have
\begin{equation}
\label{halfhalf}
\sup_{\substack{\ket{\psi}\in U \\ \|\ket\psi\|=1}} \|\ket\psi\|_{2,k}=\sup_{\|\ket\psi\|_{2,k}=1}\|P\ket\psi\|.
\end{equation}
\end{lemma}
\begin{proof}
We use estimates from the proof of Theorem \ref{convergence}. Recall that $\ket\psi\in U$ is arbitrary.
We can see that the LHS in \eqref{halfhalf} is lesser or equal than the RHS by dividing \eqref{increases} with $p=2$ by $\|\ket\psi\|_{2,k}=\sqrt{\sum^k_{i=1}(\lambda_i^\psi)^2}$. The opposite inequality is found by making the same division in \eqref{eq2} and \eqref{eq3}.
\end{proof}
An algorithm based on semi-definite programming was suggested to compute lower bounds to this quantity \cite{Johnston2}. Another approach \cite{howtocompute} is included in the MATLAB package QETLAB \cite{qetlab} that we used for tests, but comparisons suggest that it is much slower and more likely to get stuck at unwanted fixed points than the one discussed here. In many respects it also has a narrower scope, with the big exception that it can handle general operators rather than just orthogonal projections.

\section{Applications and Examples}
\label{apps}
In this section, we discuss several ways in which the algorithm can be applied. In each case, we simply ask how the algorithm can help to explore a certain topic. In some cases, we include the success rate and the number of iterations needed to get close to the maximum (Table \ref{numbers}) to give the reader an idea of how the algorithm performs in practice.

\begin{table}
\caption{The following data is included to give an indication of the number of iterations required and the prevalence of the fixed point issue, partly because it is hard to make general statements about this.
The first column in the table lists an application discussed in Section \ref{apps} for which the global maximum is either known or conjectured. We run the algorithm ten times to try to find it. The second column lists the average number of iterations needed for the norm to get within $10^{-5}$ of the conjectured or known answer if the procedure does not get stuck at other fixed points. The third column lists the number of times this does not happen and the procedure is successful.}
\begin{center}
\begin{tabular}{ l|c|c} 
 & average no.\ of iterations  & no.\ of successes \\
\hline
Conjecture \ref{Yangconj} with $d=12$, $N=8$ and $K=4$ &4.4&10\\
Conjecture \ref{myconjecture} with $d=10$, $N=6$ &17.1&10\\
Conjecture \ref{conjentropies} for $S_2$ with $d=10$, $N=6$ &9.8 &10\\
AME state (Section \ref{AME}) for $d=3$, $N=3$  &62.3 &10 \\
AME state for $d=4$, $N=4$  &234.7 &3 \\
AME state for $d=5$, $N=4$  & n/a &0\\
AME state for $d=5$, $N=5$  & 9703.5 &10
\end{tabular}
\end{center}
\label{numbers}
\end{table}

\subsection{$N$-representability and fermionic particle entanglement}
\label{fermions}
The Hilbert space of $N$ fermions---identical particles---is the antisymmetric product $\wedge^N\mathcal{H}$ of $N$ copies of the 1-particle space $\mathcal{H}$ (see Definitions \ref{Aproduct} and \ref{atp} below). States in this Hilbert space have Schmidt decompositions \eqref{Schmidt} with $\mathcal{H}_A=\wedge^K\mathcal{H}$ and $\mathcal{H}_B=\wedge^{N-K}\mathcal{H}$ and associated reduced density matrices \eqref{rdm}. The overall antisymmetry places strong constraints on both, and people have been trying to find out what these are since 1959 \cite{Coulsonchallenge,Coulson}.\footnote{We shall make no attempt to give a complete overview of the very large number of papers discussing this problem, but see for example \cite{Colemanbook}.}

What motivated them was quantum chemistry: physical interactions like the Coulomb repulsion between electrons usually act pairwise. This means that interaction operators take the form $\sum_{i<j}V_{ij}$, where each term acts as a fixed 2-body interaction $V_{12}$ on particles $i$ and $j$ and as the identity on all others. The expectation value of such an operator for an $N$-fermion state $\ket\psi$  then becomes
\begin{equation}
\label{intera}
\sum_{i<j}\bra{\psi}V_{ij}\ket\psi=\tbinom{N}{2}\Tr[\rho^\psi_{2}V],
\end{equation}
where $\rho^\psi_{2}$ is the reduction \eqref{rdm} of $\ket\psi$ to two particles (note that all reductions are equivalent because of the antisymmetry). 

Say that we wanted to minimize a quantity like \eqref{intera} by varying over all states $\ket\psi$. Computationally, we would favour the right-hand side of the equation because it involves a 2-particle space only, but there is one obvious problem: we do not know what constraints the antisymmetry imposes on the reduction $\rho^\psi_{2}$, or equivalently the Schmidt vectors and coefficients in \eqref{Schmidt}. This means that we do not know which $\rho^\psi_{2}$ to include in the minimization of the RHS of \eqref{intera}.

Enticed by the computational gains that such knowledge would bring, people actively thought about this problem in the sixties. It goes by the name of \textit{$N$-representability} \cite{Coleman} and is a specific example of the \textit{quantum marginal problem} \cite{KlyachkoQuantumMarginal}. It is mostly strongly associated with Coulson, who raised the challenge at a conference in 1959 \cite{Coulsonchallenge,Coulson}, and Coleman, who studied the problem throughout his career (see e.g.\ the papers \cite{Coleman, Coleman2}, book \cite{Colemanbook} and summary \cite{Colemanstory}). 

These days, entanglement provides a new motivation to look at it. The entanglement between $K$ fermions and the remaining fermions in the system---as associated to a Schmidt decomposition \eqref{Schmidt} with $\mathcal{H}_A=\wedge^K\mathcal{H}$, $\mathcal{H}_B=\wedge^{N-K}\mathcal{H}$---is known as \textit{particle entanglement} \cite{Schoutens}. It is sometimes confused with \textit{mode entanglement}, which describes the entanglement between restrictions of the state to complementary regions of (Hilbert) space, or even dismissed as irrelevant \cite{Benatti,Shi}, but that is certainly not the case in the context of condensed matter systems \cite{Dowling, Schoutens} and quantum chemistry \cite{Coulson}. 

So where do we stand on $N$-representability today? We know that inferring the full implications of antisymmetry is a QMA-hard problem---not even efficiently solvable by a quantum computer \cite{Christandl}. That should not discourage us from trying to extract basic mathematical  information: the full constraints imposed by antisymmetry on the Schmidt decomposition with $K=1$ were recently computed in systems with small dimensions and particle numbers \cite{Klyachko}. Interestingly, they go beyond the original Pauli principle that says that no two fermions can occupy the same state. Our algorithm cannot shed new light on this progress, and we restrict the discussion to $N$-representability questions that it can handle.\\

To summarize the above, it has long been recognized that is is important to infer the consequences of antisymmetry on the Schmidt decomposition \eqref{Schmidt} and the related reduced density matrices \eqref{rdm} because that allows for a more efficient computation of \eqref{intera}. In this section, we numerically verify a number of conjectures regarding the effect of antisymmetry on the Schmidt coefficients only. We also formulate a new one (Conjecture \ref{myconjecture}).

The following basic definitions will be used below.
\begin{definition}[Hilbert space of $N$ fermions]
\label{Aproduct}
We define $\wedge^N\mathcal{H}$---the antisymmetric tensor product of $N$ copies of a Hilbert space $\mathcal{H}$---as the subspace of $\otimes^N\mathcal{H}$ whose elements $\ket\psi$ are completely antisymmetric upon exchange of two particles. This means that for $1\leq i_1<i_2\leq N$ and $U_{(i_1\ i_2)}$ the unitary implementing the permutation $(i_1\ i_2)$ on $\otimes^N\mathcal{H}$,
\begin{equation}
U_{(i_1\ i_2)}\ket\psi=-\ket\psi.
\end{equation}
The orthogonal projection onto this space is
\begin{equation}
P:=\frac{1}{N!}\sum_{\sigma\in S_N}(-1)^{\text{sgn}(\sigma)}U_\sigma,
\end{equation}
where $S_N$ is the group of permutations of $N$ elements.
\end{definition}

\begin{definition}[Antisymmetric tensor product]
\label{atp}
Let $K_1,K_2\in\mathbb{N}$, $\ket\Psi\in\wedge^{K_1}\mathcal{H}$ and $\ket\Phi\in\wedge^{K_2}\mathcal{H}$. Let $P$ be the orthogonal projection onto the antisymmetric subspace $\wedge^{K_1+K_2}\mathcal{H}$. We define the antisymmetric tensor product
\begin{equation}
\ket{\Psi\wedge\Phi}:=\begin{cases}
\frac{P\ket{\Psi\otimes\Phi}}{\|P\ket{\Psi\otimes\Phi}\|}  & \text{if $\|P\ket{\Psi\otimes\Phi}\|\neq0$} \\
\hspace{0.6cm}$0$ & \text{otherwise}
  \end{cases}.
\end{equation}
It is easy to check that this tensor product is associative $(\alpha\wedge\beta)\wedge\gamma=\alpha\wedge(\beta\wedge\gamma)$ and so we will simply denote such products as $\alpha\wedge\beta\wedge\gamma$. Note that it is not linear.
\end{definition}

\begin{definition}[Slater determinant]
\label{Slaterdef}
An antisymmetric product of $N$ orthonormal vectors $\ket{\phi_1},\dots,\ket{\phi_N}\in\mathcal{H}$,
\begin{equation}
\label{Slaterdet}
\ket{\phi_1\wedge\dots\wedge\phi_N},
\end{equation}
is known as a Slater determinant. For $\mathcal{H}=\mathbb{C}^d$ with orthonormal basis $\ket{\phi_1},\dots,\ket{\phi_d}$, the $\binom{d}{N}$ $N$-fermion Slater determinants that we can build from this basis are an orthonormal basis of $\wedge^N\mathbb{C}^d$.
\end{definition}

We first discuss some known facts. The following result is essentially the Pauli principle.

\begin{theorem}[Coleman \cite{Coleman}]
For $\mathcal{H}_A=\mathcal{H}$ and $\mathcal{H}_B=\wedge^{N-1}\mathcal{H}$,
\begin{equation}
\sup_{\substack{\ket{\psi}\in\wedge^{N}\mathcal{H}\\ \|\ket\psi\|=1}} \|\ket\psi\|^2_{2,1}=\frac{1}{N}.
\end{equation}
Let $P$ be the orthogonal projection onto $\wedge^{N}\mathbb{C}^d$. The supremum above is attained if and only if
\begin{equation}
\ket\psi=\ket{\phi\wedge\Phi}=\frac{P\ket{\phi\otimes\Phi}}{\|P\ket{\phi\otimes\Phi}\|}
\end{equation}
for $\ket{\phi}\in\mathcal{H}$, $\ket{\Phi}\in\wedge^{N-1}\mathcal{H}$ satisfying $(\bra{\phi}\otimes\mathds{1})\ket{\Phi}=0$.
\end{theorem}

The 2-particle case is much harder, and the only clear-cut analytical result was obtained by Yang. 
\begin{theorem}[Yang \cite{Yang}, see \cite{Coleman,CLR} for simplified proofs]
\label{Yangeig}
Let $N,d$ be even. For $\mathcal{H}_A=\wedge^2\mathbb{C}^d$ and $\mathcal{H}_B=\wedge^{N-2}\mathbb{C}^d$,
\begin{equation}
\sup_{\substack{\ket{\psi}\in\wedge^{N}\mathbb{C}^d\\\|\ket\psi\|=1}} \|\ket\psi\|^2_{2,1}=\frac{1}{N-1}\frac{d-N+2}{d}.
\end{equation}
Any optimizer $\ket{\Psi^{(N)}}$ can be built from an orthonormal basis $\ket{\psi_i}$ by defining $\ket{\Psi^{(2)}}:=\sum^{d/2}_{i=1}\sqrt{2/d}\ket{\psi_{2i-1}\wedge\psi_{2i}}$ and
\begin{equation}
\label{Yang state}
\ket{\Psi^{(N)}}:=\ket{\Psi^{(2)}\wedge\dots\wedge\Psi^{(2)}}=\frac{P\ket{\Psi^{(2)}\otimes\dots\otimes\Psi^{(2)}}}{\|P\ket{\Psi^{(2)}\otimes\dots\otimes\Psi^{(2)}}\|},
\end{equation}
where $P$ is the orthogonal projection onto $\wedge^{N}\mathbb{C}^d$. Such states are known as Yang states.
\end{theorem}

No other results like this are known, but these are exactly the kind of questions that we can investigate with the algorithm. General statements of this form lead to interesting mathematical inequalities (see \cite{CLR,IKKS}), interesting states (the Yang states are examples of the BCS states that appear in superconducting systems \cite{Yang}), and they allow us to investigate the particle entanglement properties of fermions. These are some of the reasons that there appears to be a renewed interest in these questions: for example, the theorem above was recently rediscovered in the context of hard-core bosons \cite{Schilling}; a weaker version of the well-known bosonic analogue \cite{Takagi,Schmidtbosons} of Lemma \ref{YangYoula} was presented as a new result in \cite{LoFranco}; part of Conjecture \ref{conjentropies} was stated as a fact in \cite{Schoutens}; known $N$-representability constraints were put to use in quantum algorithms in \cite{Rubin}.

The conjectures below are all confirmed by the algorithm for up to 10 fermions.\footnote{See Table \ref{numbers} to get a sense of the number of iterations required and success rates.} We hope that their discussion will serve as a brief review of some interesting open questions and some useful past results.

The first conjecture generalizes Theorem \ref{Yangeig} to higher-order density matrices.

\begin{conjecture}[Yang \cite{Yang}, see \cite{CLR} for exact constants]
\label{Yangconj}
Let $N,\ d$ and $K\leq N$ be even. For $\mathcal{H}_A=\wedge^K\mathbb{C}^d$ and $\mathcal{H}_B=\wedge^{N-K}\mathbb{C}^d$, the Yang state $\ket{\Psi^{(N)}}$ is a maximizer of $\|.\|^2_{2,1}$ among the normalized vectors of $\wedge^N\mathbb{C}^d$.
For $d\to\infty$, the limiting value is $\tbinom{N/2}{K/2}/\tbinom{N}{K}$.
\end{conjecture}

Yang himself \cite{Yang2} proved a version of this conjecture that follows easily from a very useful antisymmetrization formula obtained by Sasaki \cite{Sasaki, Coleman}, but, as explained in \cite{CLR}, his result amounts to a watered-down version of the full conjecture. Finding a complete proof is profoundly more difficult than it is in the $K=2$ case, for which the following lemma is essential.
\begin{lemma}[Yang \cite{Yang}, Youla \cite{Youla}]
\label{YangYoula}
For a fermionic state $\ket\psi\in\wedge^2\mathcal{H}$, the Schmidt decomposition \eqref{Schmidt} takes the special form
\begin{equation}
\sum_i\lambda_i \sqrt{2}\ket{\psi_{2i-1}}\wedge\ket{\psi_{2i}},
\end{equation}
where the $\ket{\psi_i}$ form an orthonormal basis of $\mathcal{H}$.
\end{lemma}

\begin{remark}
Similarly, the analogue for bosonic states \cite{Schmidtbosons} in $\otimes^2_{s}\mathcal{H}$ (effectively Takagi factorization \cite{Takagi}) says that the Schmidt decomposition can always be written as a sum of terms $\lambda_i \ket{\psi_{i}}\otimes\ket{\psi_{i}}$.

For $N$ even, note that $\wedge^N\mathcal{H}\subset\otimes^2_{s}(\wedge^{N/2}\mathcal{H})$ if $N/2$ is even and $\wedge^N\mathcal{H}\subset\wedge^2(\wedge^{N/2}\mathcal{H})$ if it is odd. Therefore, the Schmidt decomposition of a state in $\wedge^N\mathcal{H}$ with $\mathcal{H}_A=\mathcal{H}_B=\wedge^{N/2}\mathcal{H}$ has a structure similar to one of those mentioned above, depending on whether $N/2$ is even or odd. This is one of the few cases in which the consequences of antisymmetry are partly evident. 
\end{remark}

Returning to the topic of Theorem \ref{Yangeig}, Lemma \ref{YangYoula} provides us with a pairing structure for one of the Schmidt vectors, which can then be used to find the optimizer \eqref{Yang state}. For general even $K$ such pairing is generically absent, and proving that the maximizer possesses it is truly the biggest challenge in overcoming Conjecture \ref{Yangconj}: if we assume pairing, arguments like \cite{Luo} get close to the full answer. Note that there is a strong intuition for pairing: it turns the fermions into hard-core bosons (in a particular basis), and these can be condensed in the same state as long as we spread out their wave functions (see \cite{Schilling}). The Yang state also has an interesting symmetry (see Theorem \ref{Yangsymmetry} below), and the fixed point equation \eqref{fixedpoint} takes on a particularly simple form---perhaps these can shed new light on this problem.

Moving on to two Schmidt coefficients, the following conjecture is easily obtained with the algorithm.
\begin{conjecture}
\label{myconjecture}
Let $N$ and $d$ be even. For $\mathcal{H}_A=\wedge^2\mathbb{C}^d$ and $\mathcal{H}_B=\wedge^{N-2}\mathbb{C}^d$,
\begin{equation}
\sup_{\substack{\ket{\psi}\in\wedge^{N}\mathbb{C}^d\\\|\ket\psi\|=1}} \|\ket\psi\|^2_{2,2}=\frac{1}{\tbinom{N}{2}}\left[1+\frac12(N-2)\frac{d-N+2}{d-2}\right].
\end{equation}
Consider a basis consisting of $\ket{\phi_1},\ket{\phi_2}$, and $\ket{\psi_i}$ with $1\leq i\leq d-2$, and a Yang state $\ket{\Psi^{(N-2)}}$ \eqref{Yang state} built from the $\ket{\psi_i}$. An optimizer is then given by
\begin{equation}
\label{2eigs}
\ket{\phi_1\wedge\phi_2\wedge\Psi^{(N-2)}}=\frac{P\ket{\phi_1\otimes\phi_2\otimes\Psi^{(N-2)}}}{\|P\ket{\phi_1\otimes\phi_2\otimes\Psi^{(N-2)}}\|},
\end{equation}
where $P$ is the orthogonal projection onto $\wedge^{N}\mathbb{C}^d$.
\end{conjecture}
Surprisingly, the limiting value for $d\to\infty$ is $(N-1)^{-1}$: exactly the same as the maximum of $\|.\|^2_{2,1}$ in that limit. This shows that eigenvalues of fermionic reduced density matrices are coupled very strongly. Coleman \cite{Coleman2} had already predicted a ($d$-independent) deviation of $O(N^{-2})$ from $O((N-1)^{-1})$, and studied a general class of states containing \eqref{Slaterdet}, \eqref{Yang state} and \eqref{2eigs} called (generalized) \textit{antisymmetric geminal powers} (AGP, \cite{Coleman2,Colemanbook}) that demonstrate similar coupling between eigenvalues. These states are seemingly very important for these kinds of problems, but the exact role they play in the fermionic subspace is still unknown. 

To study entanglement, we move on to problems involving all Schmidt coefficients. 
\begin{conjecture}[partially \cite{Schoutens},\cite{CLR},\cite{Lemm}]
\label{conjentropies}
For $\ket\psi\in\wedge^N\mathcal{H}$ normalized and $\mathcal{H}_A=\wedge^K\mathcal{H}$ and $\mathcal{H}_B=\wedge^{N-K}\mathcal{H}$, we have for the entropies \eqref{renent} and \eqref{vnent},
\begin{equation}
S_2(\rho^\psi_A), S(\rho^\psi_A)\geq \log\binom{N}{K}.
\end{equation}
That is, both entropies are minimized by reduced density matrices of Slater determinants \eqref{Slaterdet}.
\end{conjecture}

This says that Slaters have the lowest particle entanglement among all the fermionic states. That may seem like a fairly safe claim, but the mathematics behind it is not at all trivial. For instance, the 2-R\'enyi entropy conjecture is equivalent to
\begin{equation}
\Tr[(\rho^\psi_A)^2]\leq \frac{1}{\tbinom{N}{K}}.
\end{equation}
For $K$ small compared to $N$, this is $O(N^{-K})$, so that we see that a fermionic $K$-particle reduced density matrix can have a most $O(N^a)$ eigenvalues of $O(N^{-(K+a)/2})$, which gives a highly non-trivial bound for $0\leq a\leq K$. This fits with Theorem \ref{Yangeig} and Conjecture \ref{Yangconj} for $a=0$, and it provides another example of the constraints that antisymmetry puts on the Schmidt coefficients. 

Our current understanding of these effects is rather limited, but they are fundamental to the structure of fermionic states, and hence to quantum chemistry and condensed matter. We hope that this section helps to renew interest in a challenging topic that has remained challenging despite the attention it rightfully received in the sixties.

\subsection{Absolutely maximally entangled states}
\label{AME}
An $N$-partite state $\ket\psi\in\otimes^N\mathbb{C}^d$ is called \textit{absolutely maximally entangled} (AME) or a \textit{perfect tensor} if all possible reductions to $\lfloor N/2\rfloor$ parties are maximally mixed \cite{Huber2,perfecttensor}. These states have attracted some attention in the context of quantum error-correcting codes, but we ignore questions about their relevance and focus on their existence for given $N$ and $d$ instead. It turns out this greatly depends on these parameters: for $d=2$, AME states only exist for $N=2,3,5,6$; for $d=3$, they exist for $N=2,3,4,5,6,7,9,10$ and possibly 11 and 15 \cite{Hubersite}.

How can the algorithm help to decide on existence? The reduction \eqref{rdm} to $\lfloor N/2\rfloor$ parties is related to the Schmidt decomposition \eqref{Schmidt} with $\mathcal{H}_A=\otimes^{\lfloor N/2\rfloor}\mathbb{C}^d$ and $\mathcal{H}_B=\otimes^{\lceil N/2\rceil}\mathbb{C}^d$. If $\ket\psi$ is maximally mixed, all its $n_s=d^{\lfloor N/2\rfloor}$ Schmidt coefficients $(\lambda^\psi_i)^2$ equal $1/n_s$. It is easy to see that the Cauchy--Schwarz bound
\begin{equation}
\|\ket\psi\|_{1,n_s}=\sum^{n_s}_{i=1} \lambda^\psi_i \leq \sqrt{n_s}
\end{equation}
is satisfied, and only satisfied in that case. In other words, to obtain an AME state, we need to maximize $\|\ket\psi\|_{1,n_s}$ across $\binom{N}{\lfloor N/2\rfloor}$ cuts at the same time, and we can use the general form of Corollary \ref{morecuts} with $P=\mathds{1}$, $p=1$ and $k=l=n_s$. This is particularly interesting in the cases $d=6,\ N=4$ and $d=4,\ N=7$ where existence of AME states is undecided \cite{Hubersite}, but unfortunately these are out of reach---seemingly not because of computational power, but because of the dominance of other fixed points (see e.g.\ the behaviour for $d=5,\ N=4$ in Table \ref{numbers}). In any case,  the algorithm quickly finds the known AME states for $d=2$; the ones with $N\leq7$ for $d=3$; with $N\leq4$ for $d=4$, which is a good example of its ability both to maximize entanglement and to treat several cuts at the same time. The method can also be applied to similar concepts, such as \textit{locally maximally entangled} states \cite{Raamsdonk}, and mixed-dimensional AME states \cite{Huber}.\footnote{In \cite{Huber}, it was shown that mixed-dimensional AME states exist for four subsystems and dimensions 2x3x3x3, but not 2x2x2x3 and 2x2x3x3; the algorithm shows that they also exist in the (novel) cases 2x2x2x4, 2x2x3x4, 2x2x4x4, 2x3x3x4, 2x3x4x4, 2x4x4x4, 3x3x4x4, 3x4x4x4, and strongly suggests they do not exist for 3x3x3x4, covering all cases with maximal dimension 4. It can successfully be applied in some cases with higher dimensions as well.}

\subsection{Dimensions of varieties of pure states}
\label{varieties}
How large can a subspace $U\subset\mathcal{H}_A\otimes\mathcal{H}_B$ be if we do not want it to contain states that have $\leq r$ non-zero Schmidt coefficients, or in other words, \textit{Schmidt rank} $\leq r$? The answer is $(d_A-r)(d_B-r)$, as shown by Cubitt, Montanaro and Winter \cite{Cubitt}. It is possible to heuristically reach this prediction: the (affine algebraic) variety of states with Schmidt rank $\leq r$ has dimension $r(d_A+d_B-r)$ (see point 1.\ below) and so intuitively the largest dimension of a subspace that does not intersect the variety in any point other than the origin---the answer to the question---is indeed $d_Ad_B-r(d_A+d_B-r)=(d_A-r)(d_B-r)$. Such intuition can be made rigorous with standard results in algebraic geometry as both the subspace and this variety are essentially projective varieties, see \cite{Parthasarathy, Cubitt, Harris, Shafarevich}.

It may seem like an unusual application, but we could have answered this question numerically.  The states with Schmidt rank $\leq r$ are maximizers for the problem \eqref{goal} with $k=r$ and $p=2$. Therefore, if these states are present in a given subspace $U\subset\mathcal{H}_A\otimes\mathcal{H}_B$, the algorithm should find them. We now generate random $U$ of increasing dimension, and expect that they generically avoid the rank $\leq r$ states if they can. The largest dimension $D$ for which the random subspaces do not contain any rank $\leq r$ states is the predicted answer to the question.

Could there not be a special subspace with dimension $D+1$ that also avoids the states with rank $\leq r$? The relation between $D$ and the dimension of the variety makes that seem improbable: if $D+1$ can be attained, the variety's dimension is at most $d_Ad_B-(D+1)$, but then a random $(D+1)$-dimensional subspace is unlikely to contain rank $\leq r$ states---crudely think of a plane and a line that both intersect the origin in $\mathbb{R}^3$. 

Of course, these are mere heuristics, and we tested the idea on different optimizers of problems $\eqref{goal}$ for which the (rigorous) dimension-argument holds, namely Schmidt rank $\leq r$ states (maximize $\|.\|_{2,r}$), maximally entangled states (maximize $\|.\|_{1,n_s}$ as shown in the previous section), bosonic condensate states (maximize $\|.\|_{2,1}$ if we split off 1 particle), and special fermionic states such as Slater determinants \eqref{Slaterdet} (maximize $\|.\|_{2,1}$ for $K=1$) and Yang states \eqref{Yang state} (maximize $\|.\|_{2,1}$ for $K=2$). Table \ref{dimtable} lists the results obtained with the algorithm and compares them to the dimensions of the varieties as obtained with parameter counting. The former should be the dimension of the space minus the latter (possibly up to $1/2$ complex, or 1 real dimension, see below), which is indeed the case.

\begin{table}
\caption{Consider the space in the first column. The third column shows the complex dimension of the variety of states specified in the second column, based on parameter counting. The fourth column shows the maximal complex dimension that a subspace can have if it does not contain the states in the second column, as predicted by the algorithm. As expected, this is the dimension of the space minus the dimension in the third column, possibly up to $1/2$ complex, or 1 real dimension.}
\begin{center}
\small
\setlength\extrarowheight{5pt}
\begin{tabular}{ c|c|c|c} 
space & states & dim. of variety & max. dim. of subspace, $D$\\
\hline
 $\mathbb{C}^{d_A}\otimes\mathbb{C}^{d_B}$& rank $\leq r$ & $r(d_A+d_B-r)$ \cite{Cubitt} & $d_Ad_B-r(d_A+d_B-r)$\\
& max.\ entangled & $d_Ad_B-\frac{\min(d_A,d_B)^2}{2}$ & $\lfloor\frac{\min(d_A,d_B)^2}{2}\rfloor$\\
 $\otimes^2_{\text{sym}}\mathbb{C}^d$ & rank $\leq r$ &  $r(d-\tfrac{1}{2}r+\tfrac12)$ & $\tbinom{d+1}{2}-r(d-\tfrac{1}{2}r+\tfrac12)$\\ 
& max.\ entangled & $\frac{1}{2}\tbinom{d+1}{2}$ & $\lfloor\frac{1}{2}\tbinom{d+1}{2}\rfloor$\\
$\wedge^2\mathbb{C}^{d}$ & rank $\leq r$ ($r,d$ even)  & $r(d-\tfrac{1}{2}r-\tfrac12)$ & $\tbinom{d}{2}-r(d-\tfrac{1}{2}r-\tfrac12)$\\ 
& max.\ entangled ($d$ even)& $\frac{1}{2}\tbinom{d}{2}$ & $\lfloor\frac{1}{2}\tbinom{d}{2}\rfloor$\\
 $\otimes^N_{\text{sym}}\mathbb{C}^d$ & $\ket{\psi\otimes\dots\otimes\psi}$  & $d$ & $\tbinom{d+N-1}{N}-d$\\ 
 $ \wedge^N\mathbb{C}^d$ & Slater ($N\leq d$) & $N(d-N)+1$ & $\tbinom{d}{N}-N(d-N)-1$\\
 & Yang ($N+2\leq d$ even)& $\frac{1}{2}\tbinom{d}{2}$ & $\tbinom{d}{N}-\lceil\frac{1}{2}\tbinom{d}{2}\rceil$
\normalsize
\end{tabular}
\label{dimtable}
\end{center}
\end{table}

We now justify the dimensions listed in the third column of Table \ref{dimtable}. Keep in mind that the Grassmannian $G(r,d)$ is the space of $r$-planes in a $d$-dimensional space \cite{Harris,Shafarevich}, and that it has dimension $r(d-r)$.
\begin{enumerate}
\item \textit{Schmidt rank $\leq r$.} A rank $r$ matrix can be represented by an isomorphism on an $r$-dimensional space together with its $d_A-r$ dimensional kernel and the $d_B-r$ dimensional complement of the image. The last two are elements of Grassmannians, and we conclude that the total dimension is $r^2+r(d_A-r)+r(d_B-r)$. In the case of (anti)symmetric matrices, the $r\times r$ matrix is (anti)symmetric, and the Grassmannian dimension only comes in once: $\tbinom{r+1}{2}+r(d-r)$ or $\tbinom{r}{2}+r(d-r)$.

\item \textit{Maximally entangled states.} First assume that $d_A=d_B=d$, so that the state can be regarded as a square matrix. The group $U(d)\times U(d)$ parametrizes the states. However, maximally entangled states $\ket\psi$ have a symmetry $U\otimes \overline{U}\ket\psi=\ket\psi$, where $\overline{\cdot}$ denotes (entrywise) complex conjugation. That means that the (complex) dimension of the variety is $2\dim(U(d))-\dim(U(d))=d^2/2$. If, say, $d_A\leq d_B$, we apply this argument to a $d_A\times d_A$ matrix and add the dimension of the Grassmannian to obtain $d^2_A/2+d_A(d_B-d_A)=d_Ad_B-\min(d_A,d_B)^2/2$. Note that this does not seem to fit with the numerical result in the fourth column of Table \ref{dimtable}, but our method selects complex subspaces and hence only distinguishes complex dimensions. The presence of an additional real parameter would be missed, resulting in a discrepancy of half a complex dimension.

For two bosons, the states are parametrized by matrices $U(d)$ that act as $U\otimes U\ket\psi$. This restricts the symmetry from before to the orthogonal group $O(d)$ ($U\in U(d)$ with $U=\overline{U}$). The dimension of the variety is $\dim(U(d))-\dim(O(d))=\frac12d^2-\frac12\tbinom{d}{2}=\frac12\tbinom{d+1}{2}$.

For two fermions, $U(d)$ again acts as $U\otimes U\ket\psi$, but the symmetry group is less obvious, see Theorem \ref{Yangsymmetry} below. It has dimension $\frac12\tbinom{d+1}{2}$, producing the desired conclusion $\frac12\tbinom{d}{2}$.

\item \textit{Bosonic condensates.} A condensate state $\ket{\psi\otimes\dots\otimes\psi}$ is completely defined by $\ket\psi\in\mathbb{C}^d$, so the dimension is $d$.

\item \textit{Fermionic states.} A Slater determinant \eqref{Slaterdet} is simply an element of the Grassmannian, but we have to add an extra dimension for phase and norm. A Yang state \eqref{Yang state} is the fermionic equivalent of a bosonic condensate and it is fully defined by a 2-particle maximally entangled fermionic state, so we again find $\frac12\tbinom{d}{2}$.
\end{enumerate}

\begin{theorem}[Symmetries of the Yang state]
\label{Yangsymmetry}
The 2-particle Yang state \eqref{Yang state} satisfies $U\otimes U\ket{\Psi^{(2)}}=\ket{\Psi^{(2)}}$ for $U\in U(d)$ that, in the basis in which the pairing of $\ket{\Psi^{(2)}}$ is defined, have $\overline{U_{(2i-1)(2j-1)}}=U_{(2i)(2j)}$ and $\overline{U_{(2i-1)(2j)}}=-U_{(2i)(2j-1)}$ for $1\leq i,j\leq d/2$, where $\overline{\cdot}$ denotes complex conjugation. The complex dimension of this symmetry group is $\frac12\tbinom{d+1}{2}$.
\end{theorem}
\begin{proof}
It is well known that a maximally entangled state
\begin{equation}
\ket\Phi=\sum^{d/2}_{i=1}\sqrt{1/d}(\ket{u_{2i-1}\otimes v_{2i-1}}+\ket{u_{2i}\otimes v_{2i}})
\end{equation}
is invariant under $U\otimes\overline{U}\ket\Phi=\ket\Phi$, where $U$ acts on the $\ket{u}$'s and $\ket{v}$'s as defined by the index $i$. Relabelling $\ket{v_{2i-1}}=\ket{u_{2i}}$ and $\ket{v_{2i}}=-\ket{u_{2i-1}}$ to obtain $\ket{\Psi^{(2)}}$, we demand
\begin{equation}
\begin{aligned}
\overline{U}\ket{u_{2i}}=\overline{U}\ket{v_{2i-1}}&=\sum^{d/2}_{j=1}\overline{U_{(2i-1)(2j-1)}}\ket{v_{2j-1}}+\overline{U_{(2i-1)(2j)}}\ket{v_{2j}}\\
&=\sum^{d/2}_{j=1}\overline{U_{(2i-1)(2j-1)}}\ket{u_{2j}}-\overline{U_{(2i-1)(2j)}}\ket{u_{2j-1}}.
\end{aligned}
\end{equation}
At the same time, only unitaries that act on the two indistinguishable particles in the same way are allowed:
\begin{equation}
\overline{U}\ket{u_{2i}}=U\ket{u_{2i}}=\sum^{d/2}_{j=1}U_{(2i)(2j-1)}\ket{u_{2j-1}}+U_{(2i)(2j)}\ket{u_{2j}}.
\end{equation}
This gives the equations. 

We now determine the real dimension of the group and then divide by two to arrive at the complex dimension $\frac12\binom{d+1}{2}$. 

Start with any $d\times d$ matrix. This has $2d^2$ real variables. The matrix is unitary if its rows form an orthonormal basis, and we should count the number of independent equations this gives. We note however, that the rows of our desired matrices have a special, paired structure by the equations we just derived. For example, row $2i$ will automatically be normalized if row $2i-1$ is. This leaves just $d/2$ normalization equations instead of the usual $d$. Similarly, the equations imply that rows $2i-1$ and $2i$ are orthogonal, and for $i\neq j$, that orthogonality of row $2i-1$ with rows $2j-1$ and $2j$ implies orthogonality of row $2i$ with rows $2j-1$ and $2j$, so that the number of orthogonality equations is just $2\binom{d/2}{2}$ instead of the usual $\binom{d}{2}$.

Note that orthogonality gives two independent real equations, but that normalization only counts for one. This gives real dimension $2d^2-d/2-2\cdot2\binom{d/2}{2}=\tbinom{d+1}{2}$, or complex dimension $\frac12\tbinom{d+1}{2}$.
\end{proof}

Given that the heuristics work well, it could be interesting to try states with more complicated structures, such as the fermionic examples in Section \ref{fermions}. After all, the dimension of a variety says something about the structure of its states, and this is a way to numerically determine what it is. 

To give one more example, the dimension of the so-called \textit{subspace variety} is well known and easy to compute  \cite{Landsberg}, but it can also be derived with numerics and Corollary \ref{morecuts}. Because it immediately leads to a generalization of the main results of \cite{Parthasarathy,Cubitt}, we add a proof of this well-known fact.
\begin{lemma}
\label{generalization}
Consider the space $\mathbb{C}^{d_1}\otimes\dots\otimes\mathbb{C}^{d_n}$. The dimension of the subspace variety consisting of states with Schmidt ranks $\leq r_i$ for decompositions \eqref{Schmidt} with $\mathcal{H}_A=\mathbb{C}^{d_i}$, $\mathcal{H}_B=\otimes_{j\neq i}\mathbb{C}^{d_i}$ is
\begin{equation}
\label{r}
r_1r_2\dots r_n+\sum^n_{i=1}\min(r_i,\Pi_{j\neq i}r_j)(d_i-r_i).
\end{equation}
\end{lemma}
\begin{proof}
First assume that $\min(r_i,\Pi_{j\neq i}r_j)=r_i$ for all $k$. For each $i$, a state $\ket\psi$ can be regarded a map $\mathbb{C}^{d_i}\to\otimes_{j\neq i}\mathbb{C}^{d_i}$ (known as a \textit{flattening} \cite{Landsberg}) with a kernel $\Lambda_i$ that has dimension at least $d_i-r_i$. To calculate the dimension of the variety, we can assume that it is indeed $d_i-r_i$ because states with bigger kernels form a lower dimensional variety. Note that the kernels can be controlled independently by unitaries acting on the relevant Hilbert space, so that we can describe $\ket\psi$ by $(\ket\phi,\Lambda_1,\dots,\Lambda_n)$, with $\ket\phi\in\mathbb{C}^{r_1}\otimes\dots\otimes\mathbb{C}^{r_n}$ and $\mathbb{C}^{r_i}$ isomorphic to $\mathbb{C}^{d_i}/\Lambda_i$. The resulting dimension is the sum of the dimensions of the Grassmannians $G(r_i,d_i)$ and the dimension of $\mathbb{C}^{r_1}\otimes\dots\otimes\mathbb{C}^{r_n}$, namely \eqref{r} (see also the \textit{incidence correspondence} arguments used in \cite{Harris}).

In reality, the number of Schmidt vectors involved in each decomposition is $\tilde{r}_i:=\min(r_i,\Pi_{j\neq i}r_j)$, so that the general result is
\begin{equation}
\label{tilder}
\tilde{r}_1\tilde{r}_2\dots \tilde{r}_n+\sum^n_{i=1}\tilde{r}_i(d_i-\tilde{r}_i).
\end{equation}
Note that only one of the $\tilde{r}_i$'s can be different from $r_i$. Assume $\tilde{r}_1=r_2\dots r_n$. It is then easy to see that \eqref{tilder} equals \eqref{r}.
\end{proof}

This fact can be used to start solving a problem raised in \cite{Cubitt}: what can we say about states with rank restrictions across several cuts? It is of course much more difficult to say something if the kernels above are no longer independent: the simplest case of  $\mathbb{C}^{d_1}\otimes\mathbb{C}^{d_2}\otimes\mathbb{C}^{d_3}\otimes\mathbb{C}^{d_4}$ was studied in \cite{ranks}, but a lot is still unknown.\footnote{In particular, there is a very interesting conjecture about ranks in \cite{ranks}. The algorithm does not seem to disprove it.} The algorithm and the idea discussed here may be able to help.

\subsection{Minimal output R\'enyi entropy for channels}
\label{channels}
Suppose that we want to use copies of a quantum channel to transmit information. It it then useful to know whether entanglement can help, that is, whether the capacity to send classical information increases if we use entangled input states. Famously, an \textit{additivity conjecture} said that this is not the case, but it was shown to be false in \cite{Hastings}. Without discussing the details, we want to review why our algorithm can be applied to a related question: the \textit{maximal $p$-norm multiplicativity conjecture} \cite{Wernerconj}. Although this conjecture has been disproved for $p>4.79$ in \cite{Werner} and for $p>1$ in \cite{HaydenWinter}, there is a desire to find explicit counterexamples for $1<p\leq2$, and that is where the algorithm might be able to help. 

\textit{Note that the $p$ in `$p$-norm multiplicativity conjecture' will be $\alpha$ in what follows to avoid confusion with the $p$ used in the Schmidt norms \eqref{Q}.}

Since we will be working with density matrices $\rho$, we define for $\alpha\geq1$,
\begin{equation}
\|\rho\|_\alpha:=\Tr[\rho^\alpha]^{1/\alpha}.
\end{equation}
We denote the set of density matrices on $\mathcal{H}$ by $S(\mathcal{H})$.

The \textit{maximal $\alpha$-norm multiplicativity conjecture} says that for channels $\mathcal{E}:S(\mathcal{H}_S)\to S(\mathcal{H}_A)$ and $\mathcal{F}:S(\mathcal{H}_{S'})\to S(\mathcal{H}_{A'})$ and $\alpha>1$, we have
\begin{equation}
\label{pnormmult}
\sup_{\rho_{SS'}\in S(\mathcal{H}_S\otimes\mathcal{H}_{S'})}\|(\mathcal{E}\otimes\mathcal{F})(\rho_{SS'})\|_\alpha
=\left(\sup_{\rho_S\in S(\mathcal{H}_S)}\|\mathcal{E}(\rho_S)\|_\alpha\right)\left(\sup_{\rho_{S'}\in S(\mathcal{H}_S')}\|\mathcal{F}(\rho_{S'})\|_\alpha\right).
\end{equation}
By their variational characterization, it is easy to see that the norms are convex on the set of density matrices \cite{Uchiyama}, so that we can restrict the maximization to pure states. Also note that the LHS in \eqref{pnormmult} has to be bigger than or equal to the RHS. We need channels $\mathcal{E},\mathcal{F}$ with a strict inequality for a counterexample.

So how can the maximization problems in \eqref{pnormmult} can be tackled numerically? Consider $\mathcal{E}$'s Kraus operators $V_i:\mathcal{H}_S\to\mathcal{H}_A$ \cite{Kraus}, with
\begin{equation}
\mathcal{E}(\rho)=\sum^{d_Sd_A}_{i=1}V_i\rho V^\dagger_i.
\end{equation}
Given a basis $\{\ket{b_i}\}^{d_B}_{i=1}$ of an ancilla space $\mathcal{H}_B$ with $d_B=d_Sd_A$, we define an isometry $V:\mathcal{H}_S\to\mathcal{H}_A\otimes\mathcal{H}_B$ by 
\begin{equation}
V=\sum^{d_Sd_A}_{i=1}V_i\otimes\ket{b_i},
\end{equation}
so that 
\begin{equation}
\mathcal{E}(\rho)=\Tr_B\left[V\rho V^\dagger \right].
\end{equation}
The object 
\begin{equation}
VV^\dagger=\sum^{d_Sd_A}_{i,j=1}V_iV^\dagger_j\otimes \dyad{b_i}{b_j}
\end{equation}
is an orthogonal projection onto the image of $V$ in $\mathcal{H}_A\otimes\mathcal{H}_B$, and this can be used in the algorithm.
By convexity, the first problem on the RHS of \eqref{pnormmult} is
\begin{equation}
\begin{aligned}
\sup_{\rho\in S(\mathcal{H}_S)}\|\mathcal{E}(\rho)\|_\alpha&=\sup_{\substack{\ket\phi\in \mathcal{H}_S\\\|\ket\phi\|=1}}\|\mathcal{E}(\dyad{\phi}{\phi})\|_\alpha\\
&=\sup_{\substack{\ket\psi\in\image(V)\\\|\ket\psi\|=1}}\Tr[(\rho^\psi_A)^\alpha]^{1/\alpha}\\
&=\sup_{\substack{\ket\psi\in\image(V)\\\|\ket\psi\|=1}} \|\psi\|^{2}_{2\alpha,d_A},
\end{aligned}
\end{equation}
which we can try to maximize. The other two optimizations can be tackled in the same way. Note that the results directly translate to minimal output R\'enyi entropies with the definition \eqref{renent}.

The link between output entropy of channels and subspaces was exploited in \cite{HaydenWinter} to find counterexamples to \eqref{pnormmult}, but the construction uses random channels. There is an interest to find more explicit counterexamples. For $\alpha>2$, the fermionic subspace $\wedge^2\mathbb{C}^d$ with $d\to\infty$ is an easy one \cite{Grudka}; although the algorithm can verify it, it would be much more useful to test examples where a proof is hard to obtain. Note that fixed points prevent the algorithm from serving as a rigorous proof technique, but it can still help to develop an intuition about good counterexample candidates. For instance, an analysis of the entangled subspaces discussed in \cite{BrannanCollins} suggests that these are unlikely to work.

\section{Conclusion}
\label{concl}
We propose a conceptually simple algorithm to find pure states that maximize Schmidt norms within a specified subspace. It can be generalized to treat several bipartite cuts at the same time. Our approach uses only pure states and not density matrices, which limits computational costs. It also allows us to distinguish between $1<p<2$, which favours large entanglement for $k=n_s$, and $p>2$, which favours small entanglement for $k=n_s$. 

We prove convergence of the Schmidt norm under iteration, but like in other approaches we notice that there are fixed points other than the global maximum where the algorithm can get  stuck. In many cases this does not hamper investigations (see e.g.\ Table \ref{numbers}); in others, it would be good to find a way to deal with this problem, for example by adding a simulated annealing component. This is currently the best way in which the algorithm can be improved.

To motivate the reader to try the iteration in their own research, we included several possible applications. Numerical tests were done with the MATLAB package QETLAB \cite{qetlab}, and conclusions are listed below.
\begin{enumerate}
\item Physical symmetries imprint themselves on the Schmidt decomposition, and hence on entanglement properties of states. The exact consequences are often not easy to infer, but we illustrated that the algorithm can sometimes help by numerically verifying a number of long-standing conjectures for fermionic states with up to 10 fermions, and formulating a new one as well. For larger systems, we are restricted by computational capacity, although efficient programming can probably offer great improvements.

Fixed points \eqref{fixedpoint} hardly ever seem to be a problem in these examples, but the full mathematical implications of the fixed point equation---How many solutions does it have? How does the algorithm behave around these solutions?---are not known, and this is an interesting direction for further research. Indeed, there are still many open questions about fermions, and we have tried to highlight some of them.

\item The algorithm can be adapted to maximize Schmidt sums across several bipartite cuts. This allows us to search small systems for absolutely maximally entangled (AME) states or perfect tensors \cite{Huber2,perfecttensor}, or for example for locally maximally entangled states \cite{Raamsdonk}. Unfortunately, unwanted fixed points pose a greater risk for larger numbers of cuts, which means that we cannot currently discover unknown AME states. We do find new `mixed-dimensional AME states' \cite{Huber}.

\item Inspired by results of Parthasarathy \cite{Parthasarathy} and Cubitt, Montanaro and Winter \cite{Cubitt} (effectively) about dimensions of certain varieties of pure quantum states, we show that these can sometimes be determined by applying the algorithm to random subspaces, and point out that a well-known fact in algebraic geometry (Lemma \ref{generalization}) leads to a simple generalization of the main results of \cite{Parthasarathy,Cubitt}. A similar approach may help to uncover relations between ranks across all bipartite cuts. 

\item For a given quantum channel, we explain how to numerically determine the minimum output $\alpha$-R\'enyi entropy with $\alpha>1$. This means that the algorithm can investigate additivity questions involving such entropies, but progress seems to require both a good intuition about which projections to consider and a targeted numerical effort.
\end{enumerate}

We encourage the reader to give the algorithm a try: it is easy to implement, the risk of fixed points is often surprisingly small, and it is likely to be faster than approaches that involve density matrices. It would be particularly interesting to see if it can help to explore entanglement properties in spin systems and how it connects with the techniques used in the context of matrix product states.

\vskip6pt

\enlargethispage{20pt}

\dataccess{This article has no additional data.}

\competing{I have no competing interests.}

\ack{I would like to thank Benjamin B\'eri, Matthias Christandl, Nilanjana Datta, Eric Hanson, Felix Huber, Adrian Kent, Christian Majenz, Sergii Strelchuk and Micha{\l} Studzi\'{n}ski for useful suggestions.}

\funding{This work is supported by the Royal Society through a Newton International Fellowship, by Darwin College Cambridge through a Schlumberger Research Fellowship, and, by membership of the HEP group in DAMTP, supported by STFC consolidated grant ST/P000681/1.}


\end{document}